\documentclass{cccg23} 
\usepackage[T1]{fontenc}
\usepackage{lmodern}
\usepackage[utf8]{inputenc}

\usepackage{amsfonts,amsmath,amssymb}
\usepackage{microtype} 
\usepackage{cite,graphicx}

\usepackage[disable]{todonotes}

\usepackage{enumitem}
\setlist[itemize]{noitemsep, topsep=0pt}

\usepackage[capitalize]{cleveref}

\newtheorem{claim}[theorem]{Claim}
\newtheorem{observation}[theorem]{Observation}
\newtheorem{corollary}[theorem]{Corollary}

\newcommand{\dist}{\ensuremath{\operatorname{dist}}}

\title{On the complexity of embedding in graph products\thanks{Work initiated at the Workshop on Graph Product Structure Theory (BIRS21w5235) at the Banff International Research Station, November 21-26, 2021.}} \author{Therese Biedl\thanks{David R.~Cheriton School of Computer Science, University of Waterloo.   Supported by NSERC.} \and David Eppstein\thanks{Department of Computer Science, University of California, Irvine. Research supported in part by NSF grant CCF-2212129.} \and Torsten Ueckerdt\thanks{Institute of Theoretical Informatics, Karlsruhe Institute of Technology}} 
\date{}

\begin{document}

\maketitle

\begin{abstract}
Graph embedding, especially as a subgraph of a grid, is an old topic in VLSI design and graph drawing.   In this paper, we investigate related questions relating to the complexity of embedding a graph~$G$ in a host graph that is the strong product of a path~$P$ with a graph~$H$ that satisfies some properties, such as having small treewidth, pathwidth or tree depth.   We show that this is NP-hard, even under numerous restrictions on both $G$~and~$H$.   In particular, computing the row pathwidth and the row treedepth is NP-hard even for a tree of small pathwidth, while computing the row treewidth is NP-hard even for series-parallel graphs.
\end{abstract}

\section{Introduction}

Layered treewidth, layered pathwidth, 
and row treewidth 
are structural parameters of graphs that have played a central role in recent developments in graph product structure theory.
(The original graph product structure theorem was proved by Dujmov\'ic et al.~\cite{DJMMUW20}; see also \cite{UWY22,BMO22,DHJLW21}
for improvements and related results.)
%
Testing whether a graph has layered pathwidth $\le 1$ is NP-complete~\cite{BanDevDuj-Algo-19}. In this work we ask analogous questions about the computational complexity of the \emph{row treewidth} of a graph $G$, the minimum possible treewidth of a graph~$H$ such that $G$ is a subgraph of the strong product $H{\boxtimes} P_\infty$ where $P_\infty$ is a 1-way infinite path:

\begin{itemize}
\item Is it NP-hard to compute the row treewidth? 
\item Is it NP-hard to test whether a planar graph has row treewidth 1, its smallest nontrivial value?
\item How complicated must $G$ be for these problems to be hard? Are they easier for planar graphs? 
\end{itemize}

Row treewidth can be naturally generalized to other product forms for $H$.  For example, the \emph{row pathwidth} of a graph $G$ is the smallest possible pathwidth of a graph~$H$ such that $G$ is a subgraph of $H{\boxtimes} P_\infty$, and similarly one can define the \emph{row treedepth} or \emph{row simple treewidth} or \emph{row simple pathwidth}.   The above questions could be asked for any of these parameters.

These questions have a geometric flavor coming from the grid-like graph products they concern. They are special cases of subgraph isomorphism, which is hard even under strong restrictions on both $G$ and the host graph~\cite{MT92}.
Our answers are that these problems are indeed hard, even for very simple graphs. It is NP-hard to test whether a tree of bounded pathwidth has row pathwidth one, %
and the same holds for row simple pathwidth and row treedepth.
Row treewidth is trivial for trees, but it is NP-hard to test whether series-parallel graphs of bounded degree and bounded pathwidth have row treewidth one. Under the small set expansion conjecture (a strengthening of the unique games conjecture from computational complexity theory), row treewidth, row pathwidth, layered treewidth, and layered pathwidth are hard to approximate with constant approximation ratio. We provide a few positive results as well:   Testing embeddability in $P{\boxtimes} P$ (a grid with diagonals) is polynomial for caterpillars, and testing embeddability in $P\Box P$ (a grid) is polynomial for planar graphs of bounded treewidth and bounded face size.


\subsection{Definitions}

A \emph{tree decomposition} of a graph $G$ is a tree $T$ whose vertices are labeled with subsets of vertices of $G$, called \emph{bags}. Each vertex must belong to bags forming a connected subtree of $T$, and each edge of $G$ must have both endpoints included together in at least one bag. If $T$ is a path, it forms a \emph{path decomposition}. The \emph{width} of the decomposition is the size of the largest bag, minus one. The \emph{treewidth} of $G$ is the smallest width of a tree decomposition of $G$, and the \emph{pathwidth} is the smallest width of a path decomposition. 
A tree decomposition is \emph{$w$-simple} if each set of $w$ vertices belongs to at most two bags. The \emph{simple pathwidth [simple treewidth]} of $G$ is the smallest $w$ such that $G$ has a $w$-simple path [tree] decomposition of width~$\le w$.    The \emph{treedepth} of a graph $G$ is the smallest height of a rooted tree $T$ on the vertices of $G$ such that every edge of $G$ connects an ancestor-descendant pair in~$T$.

For connected graphs with at least one edge, these width parameters have minimum value one. The graphs with treewidth one are trees. The graphs with treewidth two are the \emph{series-parallel graphs} and their subgraphs. The graphs with pathwidth one are not just paths, but \emph{caterpillars}: trees whose non-leaf vertices form a path called the \emph{spine}.  (The leaves of a caterpillar are called \emph{legs}.) The graphs with simple pathwidth one are paths. The graphs with treedepth one are \emph{stars}, graphs $K_{1,\ell}$ for integer $\ell$.
To avoid having to specify a specific number of vertices, it is convenient to let  $P_\infty=\langle p_0,p_1,\dots\rangle$ 
denote a \emph{ray}, a one-way infinite path, to let $C_\infty$ denote the caterpillar with infinite-length spine and infinitely many legs at each spine-vertex, and to let $S_\infty$ be a star with infinitely many degree-1 vertices.

The \emph{strong product} of two graphs $G{\boxtimes} H$ has a vertex $(u_i,v_j)$ for each pair of a vertex $u_i$ in $G$ and a vertex $v_j$ in $H$, and an edge connecting two pairs $(u_i,v_j)$ and $(u_{i'},v_{j'})$ when $u_i$ and $u_{i'}$ are either adjacent in $G$ or identical, and $v_i$ and $v_{i'}$ are either adjacent in $H$ or identical. For instance, the strong product of two paths is a \emph{king's graph}, the graph of moves of a chess king on a chessboard whose rows and columns are indexed by the vertices of the paths (see also \cref{fig:products}).

A \emph{layering} of a graph $G$ is a partition of the vertices into sets $L_0,L_1,\dots$ such that for any edge the endpoints are in the same or in consecutive sets. It can be understood as a representation of a graph $G$ as a subgraph of $P{\boxtimes} K$ for a path $P$ and a complete graph $K$; the layers of the layering are the subsets of vertices in this product coming from the same vertex of the path. Layered tree decompositions and path decompositions of a graph consist of a tree or path decomposition of the graph, together with a 
layering. Their width is the size of the largest intersection of a bag with a layer, minus one. The layered treewidth~\cite{DujMorWoo-JCTB,Sha-EuroCG-13} or layered pathwidth~\cite{BanDevDuj-Algo-19} of $G$ is the minimum width of such a decomposition. Instead, the \emph{row treewidth} or \emph{row pathwidth} of $G$ is the minimum treewidth or pathwidth of a graph $H$ for which $G$ is a subgraph of $P{\boxtimes} H$ for some path $P$. Intuitively, row treewidth and row pathwidth restrict the notion of layered treewidth and layered pathwidth by requiring each layer to have the same decomposition. These concepts are not equivalent: the layered treewidth of any graph $G$ is at most its row treewidth plus one, but there exist graphs with layered treewidth one and arbitrarily large row treewidth. A similar separation occurs also between layered pathwidth and row pathwidth~\cite{BosDujJav-DMTCS-22}.   We can similarly define layered simple treewidth/simple pathwidth/treedepth and row simple treewidth/simple pathwidth/treedepth; to our knowledge these parameters have not been studied previously.

\begin{figure*}[ht]
\hspace*{\fill}
\includegraphics[scale=0.7,page=4]{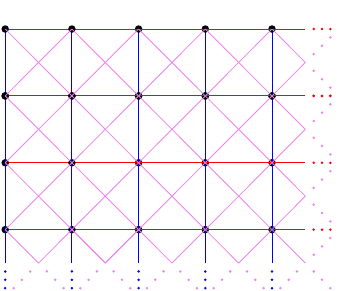}
\hspace*{\fill}
\includegraphics[scale=0.7,page=11]{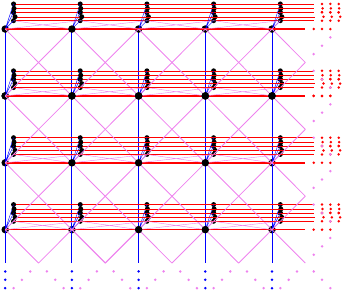}
\hspace*{\fill}
\includegraphics[scale=0.7,page=12]{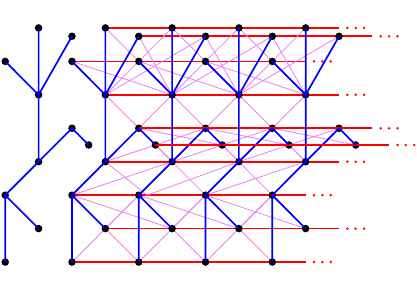}
\hspace*{\fill}
\caption{The graphs $P_\infty{\boxtimes} P_\infty$, $C_\infty {\boxtimes} P_\infty$, and $T{\boxtimes} P_\infty$ for a tree $T$.}
\label{fig:products}
\end{figure*}

We show that the following problems are NP-hard:
\begin{itemize}
\item {\sc RowSimplePathwidth}: Given a graph $G$ and an integer $k$, does $G$ have row simple pathwidth at most $k$?   We will show that this
	is NP-hard even for $k=1$, where it becomes the question whether $G$ is a subgraph of $P_\infty {\boxtimes} P_\infty$, i.e., the king's graph,
	which is why we
	also call this problem {\sc KingGraphEmbedding}.     See \cref{sec:grid} and \cref{sec:grid-details}.
\item {\sc RowPathwidth}: Given a graph $G$ and an integer $k$, does $G$ have row pathwidth at most $k$?   We will show that this
	is NP-hard even for $k=1$, where it becomes the question whether $G$ is a subgraph of $C_\infty {\boxtimes} P_\infty$. 
	See \cref{sec:rowPW}.
\item {\sc RowTreewidth}: Given a graph $G$ and an integer $k$, does $G$ have row treewidth at most $k$? We will show that this
	is NP-hard even for $k=1$, where it becomes the question whether $G$ is a subgraph of $T {\boxtimes} P_\infty$ for some tree $T$.
	See \cref{sec:rowTW} and \cref{sec:rowTW-details}.
\item {\sc RowTreedepth}: Given a graph $G$ and an integer $k$, does $G$ have row treedepth at most $k$?   We will show that this
	is NP-hard even for $k=1$, where it becomes the question whether $G$ is a subgraph of $S_\infty {\boxtimes} P_\infty$. 
	See \cref{sec:rowTD}.
\end{itemize}

It is helpful to introduce some notation for the strong product $H{\boxtimes} P_\infty$.
Recall that $P_\infty$ is a ray $\langle p_0,p_1,\dots\rangle$.
For any vertex $v\in H$, the \emph{$P$-extension} is the set of vertices $\langle v\times p_0,v\times p_1,\dots\rangle$.
For any vertex $v\times p_i \in H\times P_\infty$, the \emph{$H$-projection} is the vertex $v$ and the \emph{$P$-projection}
is the vertex $p_i$.   These concepts naturally expand to edges and paths.
Inspired by the case $H=P_\infty$ (where $H{\boxtimes}P_\infty$ is the king's graph) we define the following \emph{edge-orientations}: 
An edge $vw$ of $H{\boxtimes} P_\infty$ is \emph{horizontal} if $v,w$ have the same $H$-projection, \emph{vertical} if  
they have the same $P$-projection, and \emph{diagonal} otherwise.   Every vertex has only two incident horizontal edges.

We will occasionally also study the \emph{Cartesion product} $H\Box P_\infty$ of two graphs, which is the same as
the strong product except diagonals are omitted.   In particular, $P_\infty \Box P_\infty$ is the rectangular grid.

\section{Grid embeddings}
\label{sec:grid}

In this section we study {\sc KingGraphEmbedding}.   This problem is closely related to {\sc GridEmbedding}, the question whether a given graph $G$ is a subgraph of 
$P_\infty \Box P_\infty$.
{\sc GridEmbedding} is old and well-studied since at least the 1980s due to its connections to VLSI design.   Bhatt and Cosmadakis showed
in 1987 \cite{BC87} that {\sc GridEmbedding} is NP-hard 
even for trees of pathwidth 3 (the pathwidth was not studied explicitly by the authors, but can be verified from the construction).  Gregori \cite{Gregori89} expands their proof to binary trees.  Both proofs use a technique later called the ``logic engine'' by Eades and Whitesides \cite{EW96}.
Recently, Gupta et al.~\cite{GSZ21} strengthened the result to trees of pathwidth~2.


\begin{theorem}[Gupta et al.~\cite{GSZ21}]
\label{thm:gridEmbeddingNPhard}
    {\sc GridEmbedding} is NP-hard even for a tree of pathwidth 2.
\end{theorem}

The reductions in~\cite{GSZ21,BC87} 
can be modified to work for {\sc KingGraphEmbedding}.
Even easier is to use the following general-purpose transformation.

\begin{figure}[ht]
    \centering
    \includegraphics{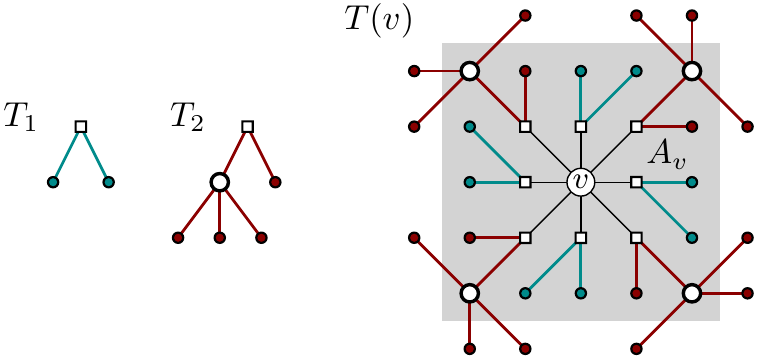}
    \caption{The trees $T_1$, $T_2$, and $T(v)$ in \cref{obs:BoxVsBoxtimes}.}
    \label{fig:grid-graph-reduction}
\end{figure}

Define $T_1$ and $T_2$ to be the trees shown in \cref{fig:grid-graph-reduction}, formed by subdividing one edge of $K_{1,1}$ and $K_{1,4}$ respectively. For a given vertex $v$, define $T(v)$ be a tree rooted at~$v$ with eight children: four copies each of $T_1$ and $T_2$, connected to $v$ at their degree-$2$ vertices.
The following is not hard to verify (see \cref{sec:grid-details}):

\begin{observation}
\label{obs:BoxVsBoxtimes}
    Let $G$ be a simple graph. Form $G'$ by replacing each vertex $v$ in $G$ by a new tree $T(v)$, and connecting a degree-$4$ vertex in $T(u)$ with a degree-$4$ vertex in $T(v)$ for each edge $uv$ in $G$.
    Then $G \subset P_\infty\Box P_\infty$ if and only if $G' \subset P_\infty{\boxtimes} P_\infty$.
\end{observation}

\begin{figure}
    \centering
    \includegraphics{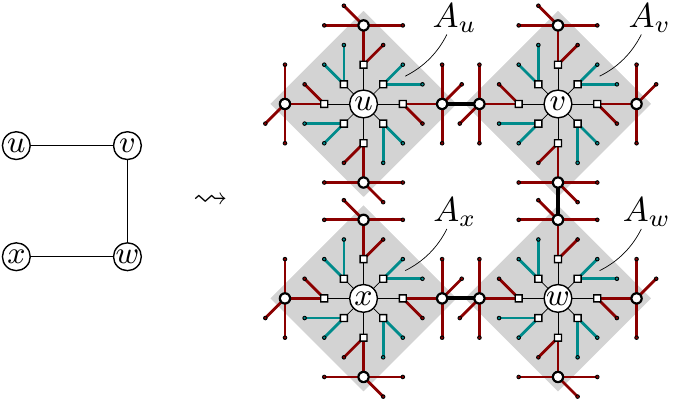}
    \caption{If $G \subset P_\infty \Box P_\infty$, then $G' \subset P_\infty {\boxtimes} P_\infty$.   (We show a 45$^\circ$ rotation of $P_\infty{\boxtimes} P_\infty$.)}
    \label{fig:grid-graph}
\end{figure}

The transformation clearly maintains a tree.
Replacing each vertex $v$ by a tree $T(v)$ of radius~$3$ increases the pathwidth by at most $3$, so applying the transformation to the tree of Gupta et al.~\cite{GSZ21} gives the following.
%


\begin{corollary}
\label{cor:gridEmbedding}
{\sc KingGraphEmbedding} is NP-hard, even for a tree of pathwidth at most $5$.
\end{corollary}

In fact, one can easily adapt the reduction of Gupta et al.~\cite{GSZ21} to show NP-hardness of {\sc KingGraphEmbedding} even for a tree of pathwidth $2$; see \cref{fig:pathwidth-2-reduction} in the appendix for an illustration.

\subsection{Caterpillars}

On the other hand, for pathwidth 1 (i.e., caterpillars), we can solve {\sc KingGraphEmbedding} 
in linear time.

\begin{theorem}\label{thm:caterpillar}
    For any caterpillar $G$ the following are equivalent.
    \begin{enumerate}[label = (\arabic*), noitemsep, topsep = 0pt]
        \item $G \subset P_\infty{\boxtimes} P_\infty$
	\label{enum:caterpillar-in-grid}
        
        \item $G$ can be embedded in $P_\infty {\boxtimes} P_\infty$ such that all spine edges are diagonal
\label{enum:caterpillar-diagonal-spine}

        \item For every subpath $Q$ of the spine of $G$ we have $\sum_{v \in V(Q)} \deg(v) \leq 6|V(Q)|+2$.\label{enum:caterpillar-degree-condition}
    \end{enumerate}
\end{theorem}

\begin{proof}
    \ref{enum:caterpillar-in-grid}${\Longrightarrow}$ \ref{enum:caterpillar-degree-condition}: \quad 
    Assume that $G$ is a caterpillar that is a subgraph of $H = P_\infty {\boxtimes} P_\infty$.
    Let $Q$ be any fixed subpath of the spine of $G$.
    Clearly, for any vertex $v \in P_\infty {\boxtimes} P_\infty$ we have $|N_H(v)| \leq 8$ and for any edge $uv \in P_\infty {\boxtimes} P_\infty$ we have $|N_H(u) \cap N_H(v)| \geq 2$.
    As caterpillar $G$ contains no triangles, for any two adjacent vertices $x,y$ in $G$ we have $N_G(x) \cap N_G(y) = \emptyset$.
    Hence
    \[
        \sum_{v \in V(Q)} \deg(v) \leq 8|V(Q)| - 2|E(Q)| = 6|V(Q)| + 2.
    \]

    
    \ref{enum:caterpillar-degree-condition} ${\Longrightarrow}$ \ref{enum:caterpillar-diagonal-spine}: \quad
    Assume that $G = (V,E)$ is a caterpillar, say with spine $\langle v_1,\ldots,v_k\rangle$.
The vertices of $P_\infty {\boxtimes}P_\infty$ naturally corresponds $\mathbb{N} \times \mathbb{N}$, where $(p_i,p_j')$ 
is mapped to $(i,j)$.
    We embed the spine of $G$ along the main diagonal, i.e., place $v_i$ at $(i,i)$ for $i=1,\ldots,k$.  
    Then, we proceed along the spine from $v_1$ to $v_k$, always placing the next leg at $v_i$ at the positions $(x,y)$ adjacent to $(i,i)$ with $x + y$ as small as possible.
    Let us say that $v_i$ is \emph{free} if 
    two leaves at $v_i$ are embedded successfully at $(i-1,i)$ and $(i,i-1)$, respectively.
    In particular, the first vertex with degree at least $4$ is always free.   (We assume there exists such a vertex, otherwise $G$ clearly can be embedded.) 

Assume that this embedding procedure fails to find a suitable position for a leaf at $v_j$ for some $j \in [k]$.
    Let $i \leq j$ be the largest index such that $v_i$ is free, and $Q = \langle v_i,\ldots,v_j\rangle$ be the subpath of the spine of $G$ from $v_i$ to $v_j$.
    Observe that $\deg_G(v_i),\deg_G(v_j) \geq 4$ and further that for $A = \{(i,i), \ldots, (j,j)\}$, there is a vertex in $V(Q) \cup N_G(Q)$ on each of the $5(j-i)+9$ points in $A \cup N(A)$ in $P_\infty {\boxtimes} P_\infty$.
    With $|V(Q) \cup N_G(Q)| = \sum_{v \in V(Q)} \deg(v) - (j-i) + 1$, it follows that 
    \begin{align*}
        & \hspace{4.6em} |V(Q) \cup N_G(Q)| > |A \cup N(A)|\\
        \Leftrightarrow & \sum_{v \in V(Q)} \deg(v) - (j-i) + 1 > 5(j-i) + 9\\
        \Leftrightarrow& \sum_{v \in V(Q)} \deg(v) > 6(j-i)+8 = 6(j-i+1) + 2\\
        & \hspace{12.5em} = 6|V(Q)|+2,
    \end{align*}
    which implies that $G$ does not satisfy~\ref{enum:caterpillar-degree-condition}.
    
    \medskip

    \ref{enum:caterpillar-diagonal-spine} ${\Longrightarrow}$ \ref{enum:caterpillar-in-grid}: \quad 
    This is immediate.
\end{proof}

\begin{corollary}
    {\sc KingGraphEmbedding} can be solved in linear time for $n$-vertex caterpillars.
\end{corollary}

\begin{proof}
    Let $G$ be a caterpillar with spine $\langle v_1,\ldots,v_k\rangle$, $k \leq n$.
    Using \ref{enum:caterpillar-degree-condition} in \cref{thm:caterpillar}, $G$ admits \emph{no} embedding into $P_\infty {\boxtimes} P_\infty$ if and only if the sequence $(\deg(v_i) - 6)_{i \in [k]}$ has a contiguous subsequence whose sum is at least~$3$.
    Finding such a subsequence is the {\sc MaximumSubarray} problem and can be solved in time $O(k)$~\cite{Gri82}.
\end{proof}

\section{Row pathwidth}
\label{sec:rowPW}

Now we consider the row pathwidth, and show that testing whether the row pathwidth is 1 is NP-hard.
This is the same as asking whether a given graph $G$ is a subgraph of $C_\infty {\boxtimes} P_\infty$.
We also consider the related problem of embedding in $C_\infty \Box P_\infty$.    Both problems are easily
shown NP-hard using another observation concerning how graph transformations affect embeddability.

\begin{observation}
\label{obs:PVsC}
\label{obs:CvsP}
Let $G$ be a simple graph, and for $k\in \{4,6\}$ let $G_k'$ be the result of adding (at any original vertex $v$ of $G$) $\max\{0,k-\deg(v)\}$ leaves that are adjacent to $v$.     Then
\begin{itemize}
\item $G \subset P_\infty \Box P_\infty$ if and only if $G'_4 \subset C_\infty \Box P_\infty$.
\item $G \subset P_\infty {\boxtimes} P_\infty$ if and only if $G'_6 \subset C_\infty {\boxtimes} P_\infty$.
\end{itemize}
\end{observation}
\begin{proof}
The forward direction is obvious: If $G$ is such a subgraph, then take the embedding of $G$ in the grid, and use the $P$-extensions of $k$ legs at each spine-vertex of $C_\infty$ to place the added leaves at each vertex $v$.

For the other direction, observe that 
all vertices on $P$-extensions of legs of $C_\infty$ have degree at most 3 in $C_\infty \Box P_\infty$, and degree at most 5 in $C_\infty {\boxtimes} P_\infty$.   We constructed $G'_k$ (for $k\in \{4,6\}$) such that the vertices of $G$ have degree $k$,
so they must be placed on the $P$-extension of a spine-vertex.   If we set $\pi$ to be the spine of $C_\infty$, therefore $G$ is embedded in $\pi \Box P_\infty$ (respectively $\pi {\boxtimes} P_\infty$) as desired.
\end{proof}

\begin{theorem}
\label{thm:row pathwidth}
It is NP-hard to test whether a tree is a subgraph of $C_\infty {\boxtimes} P_\infty$.
It is also NP-hard to test whether a tree is a subgraph of $C_\infty \Box P_\infty$.
Both results hold even for trees with constant maximum degree and pathwidth~3.
\end{theorem}
\begin{proof}
By the discussion after \cref{cor:gridEmbedding},
    testing whether $G\subset P_\infty {\boxtimes} P_\infty$ is NP-hard, even for a tree with pathwidth $2$.
    Convert $G$ into $G'$ using \cref{obs:PVsC} with $k=6$.
    This preserves a tree, increases the pathwidth by at most~1, and the maximum degree is 8.
    Also $G \subset P_\infty {\boxtimes} P_\infty$ if and only if $G' \subset C_\infty {\boxtimes} P_\infty$, which proves the first claim.   
    The second claim is similar, using 
Theorem~\ref{thm:gridEmbeddingNPhard}
and \cref{obs:PVsC} with $k=4$.
\end{proof}

\begin{corollary}
{\sc RowPathwidth}$(G)$ is NP-hard, even for trees of bounded degree and pathwidth, and even if we only want to know whether the row pathwidth is 1.
\end{corollary}

\section{Row treewidth}
\label{sec:rowTW}

We now sketch why computing row treewidth NP-hard, even for testing whether it is 1, i.e., whether a given graph can be embedded in $T{\boxtimes} P_\infty$ for some tree $T$.   (The full proof is in the appendix.)

\begin{theorem}
\label{thm:row-treewidth}
It is NP-hard to test whether a graph $G$ is a subgraph of $T{\boxtimes} P_\infty$ for some tree $T$,  even for a series-parallel graph $G$.
\end{theorem}

Our reduction from NAE-3SAT uses the \emph{logic engine} of Eades and Whitesides \cite{EW96}.   
Fix an instance $I$ of NAE-3SAT; we assume that one clause is $x_n\vee \overline{x_n}$ because we can add this without affecting existence of a solution.  We first construct a graph $G_0$ and designate some edges as horizontal/vertical.
(Figure~\ref{fig:SP_redux_1} in the appendix shows $G_0$, while Figure~\ref{fig:SP_redux_3} shows the graph derived from it.)
Start with the \emph{frame} (orange) which consists of three paths connecting two vertices $t,b$; the \emph{middle path} has 
$H:= m+2n+1$ vertical edges, while the two \emph{outer paths} have $2n$ horizontal, $H$ vertical, and then another $2n$ horizontal edges each.
Next add the \emph{armature of $x_i$} (light/dark cyan) for each variable $x_i$, which consists of two paths that attach at the vertices of the middle path at distance $i$ from $t$ and $b$.    The paths are assigned to literals $x_i$ and $\overline{x_i}$ and consist of $2n+1-2i$ horizontal edges at both ends with $H-2i$ vertical edges inbetween.
The middle $m$ rows of our drawing are called the {\em clause-rows} and assigned to one clause each.   
Finally we attach \emph{flags} (green).   Namely, at the vertex where the armature of literal $\ell_i$ intersects the row of $c_j$, we attach a leaf (via a horizontal edge) if and only if $\ell_i$ does \emph{not} occur in $c_j$.     This finishes the construction of $G_0$.

\begin{figure*}[ht]
\hspace*{\fill}
\includegraphics[scale=0.8,page=2]{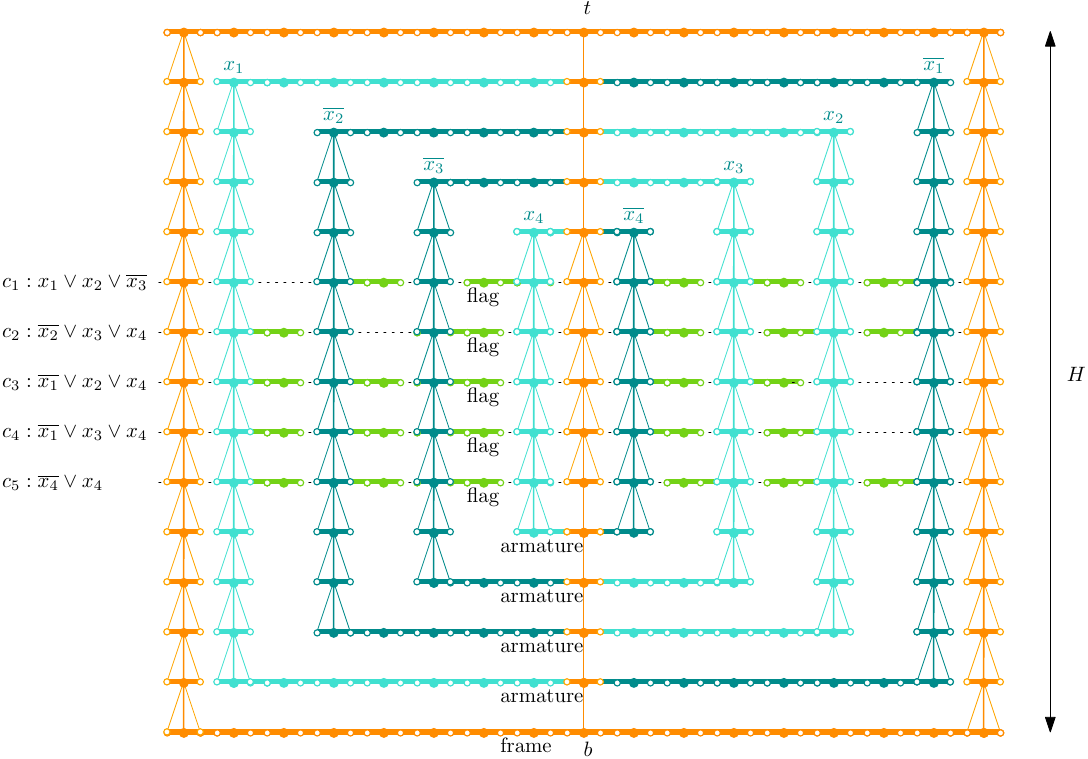}
\hspace*{\fill}
\caption{The reduction for row-treewidth. Bold edges indicate an attached $K_{2,5}$. Vertices of $G_0$ are solid.}  
\label{fig:SP_redux_3}
\label{fig:SP_redux_2}
\end{figure*}

Next we add more vertices and edges that force edge-orientations to be what we specified for $G_0$.
First, ``triple the width'': insert a new column
before and after every column that we had in our drawing of $G_0$, subdivide each horizontal edge of $G$,
and for every vertex $v$ with $k$ incident horizontal edges, add $2{-}k$ new leaves connected via horizontal edges.
(New vertices are hollow in \cref{fig:SP_redux_3}.) 
Next \emph{add an arrow-head} at some vertical edges $vw$.  Assuming $v$ is below $w$, this means adding the edges 
 $(v_\ell,w)$ and $(v_r,w)$, where $v_\ell,v_r$ are the two neighbours of $v$ adjacent to it via horizontal edges.    
We add arrow-heads at any vertical edge for which the lower endpoint $v$ does {\em not} belong to an armature. 
Call the result $G'$.
Finally we turn $G'$ into $G$ by \emph{adding a $K_{2,5}$} at every horizontal edge $uv$, i.e., adding
five new vertices that are adjacent to both $u$ and $v$.   
(To avoid clutter we do not show $K_{2,5}$ in \cref{fig:SP_redux_3}, but indicate it with a bold edge.)
Call the resulting graph $G$,
and verify that it is indeed a series-parallel graph.

One can argue (see the appendix) that if $G$ is embedded in $T{\boxtimes} P_\infty$ for some tree $T$, then all
edges with attached $K_{2,5}$ must be horizontal.   This in turn forces that $G'$ is actually embedded within
$P_\infty {\boxtimes} P_\infty$ (this is the hardest part).   The arrow-heads force the edges at which they are
attached to be vertical,
and with a counting-argument therefore the embdding of $G'$ implies an embedding of $G_0$ in $P_\infty \Box P_\infty$$P_\infty \Box P_\infty$
where the designated orientations are respected.   This is (with the standard logic engine argument) easily seen to
be equivalent to the NAE-3SAT instance having a solution.

The graph in our construction has maximum degree 16 and pathwidth $O(1)$, so computing the row treewidth remains
NP-hard even if we restrict the maximum degree or the pathwidth.

\begin{corollary}
{\sc RowTreewidth} is NP-hard, even for series-parallel graphs of bounded degree and pathwidth, even if we only want to know whether the row treewidth is 1.
\end{corollary}

An similar construction shows that testing whether $G \subset T\Box P_\infty$ for some tree $T$ is also NP-hard.   Namely,
use the same construction ($G_0$ to $G'$ to $G$), except omit the diagonal edges and replace `$K_{2,5}$' by `three paths of length 2'.
This forces all `horizontal' edges to have the desired orientation in any embedding of $G$ in $T\Box P_\infty$.   Argue
as above that then $G$ lies within $\pi\Box P_\infty$ for a path $\pi$.   Therefore any `vertical' edge $uv$ must have this
orientation, because both $u,v$ have two incident horizontal edges.   So this gives an embedding of $G_0$ in the grid that respects
the given orientation, hence a solution to NAE-3SAT.

\section{Inapproximability}

It is not known whether the treewidth or pathwidth of a graph may be approximated to within a constant factor in polynomial time, but the impossibility of doing so is known to follow from a standard assumption in computational complexity theory, the small set expansion conjecture~\cite{WuAusPit-JAIR-14}, and the best approximation ratio known for a polynomial-time approximation algorithm for the treewidth is $O(\sqrt{\log w})$, where $w$ is the treewidth~\cite{FeiHajLee-SICOMP-08}.

As we now show, the same hardness results apply to the approximation of row treewidth and row pathwidth:

\begin{theorem}
If there exists an approximation algorithm for row treewidth, row pathwidth, layered treewidth, or layered pathwidth with approximation ratio $\rho$, then there exists an approximation algorithm for treewidth or pathwidth (respectively) with approximation ratio at most $3\rho$. As a consequence, the small set expansion conjecture implies that $\rho$ cannot be $O(1)$.
\end{theorem}

\begin{proof}
Let $G$ be a graph for which we wish to approximate the treewidth or pathwidth, let $w$ be its treewidth or pathwidth, and form graph $G^+$ with treewidth or pathwidth $w+1$ by adding a universal vertex to $G$. The universal vertex forces every layering of $G^+$ to use at most three layers. $G^+$ has a trivial layering with one layer and row treewidth or row pathwidth $w+1$. Any other layering has row treewidth, row pathwidth, layered treewidth, or layered pathwidth at least $w+1/3$, because it gives a tree decomposition for $G^+$ with bags that are the unions of bags in three layers. Therefore, any approximation for the row treewidth, row pathwidth, layered treewidth, or layered pathwidth of $G^+$ gives an approximation for the treewidth or pathwidth of $G^+$, and therefore of $G$, with approximation ratio increased by at most a factor of three.
\end{proof}

Note that the constructed graph $G^+$ is not necessarily planar.   In fact, for planar graphs there are $O(1)$-approximation algorithms for the treewidth \cite{KT16}.

\section{Outlook}

In this paper, we proved that computing graph parameters such as the row pathwidth and row treewidth are NP-hard to compute, even under strong restrictions on the input graph. In fact, most of these restrictions rule out hopes for fixed-parameter tractability (or at least the possibility of finding polynomial-time algorithms in special situations). 
%
%
We do state here a few possibilities of situations where finding an embedding may be polynomial, but this mostly remains for future work:

\begin{itemize}
\item Give a graph with bounded radius, is it possible to solve {\sc RowTreewidth} or {\sc RowPathwidth} in polynomial time?  In all our hardness constructions, the graph had radius $\Theta(n)$. Bounded radius forces any layering to use a bounded number of rows, so if the row treewidth or row pathwidth is also bounded, then the treewidth or pathwidth of the original graph must also be bounded, but it is not obvious how to take advantage of this in an algorithm.  

Note that {\sc GridEmbedding} \emph{is} polynomial for graphs of bounded radius, because a graph can be embedded in a grid only if it has bounded maximum degree, and together with bounded radius this would imply bounded size, hence a constant-time algorithm.
\item For the results for {\sc GridEmbedding} (\hspace{0.1pt}\cite{BC87} and our construction in Claim~\ref{cl:constrainedEmbedding}), we very much needed the ability to change the embedding of the graph, so that we could flip armatures and flags.    What is the status if the embedding is fixed?
In particular, is testing whether a tree can be embedded in a grid NP-hard if the embedding of the tree is fixed, possibly similar to the results in~\cite{ALP22}?   
\end{itemize}

One could also ask for results for planar graphs with a fixed embedding where faces have small degrees, for example triangulated planar graphs.   In all our constructions, some faces have degree $\Theta(n)$.
Can we solve any of the problems (but especially {\sc KingGraphEmbedding}) for 
triangulated planar graphs?
This remains open, but we can make some progress if additionally also the treewidth is small.


\begin{theorem}
Let $G$ be a planar graph with treewidth $t$ and a planar drawing $\Gamma$ where all faces have degree at most $\Delta$.   Then we can test whether $G$ can be embedded in the grid (in a way that respects embedding $\Gamma$) in time $O^*(n^{3(t{+}1)\Delta})$, i.e., in polynomial time if $t\cdot \Delta\in O(1)$.
\end{theorem}
\begin{proof}
In 2013, the first author and Vatshelle \cite{BV13} studied the {\sc PointSetEmbedding} problem, where we are give a set of points $S$ and a planar graph $G$, and we ask whether $G$ has a planar straight-line drawing where all vertices are placed at points of $S$.   They showed that if $G$ has treewidth at most $t$ and face-degree at most $\Delta$, then {\sc PointSetEmbedding} can be solved in $O^*(|S|^{1.5(t{+}1)\Delta})$ time.   Their approach is to use a so-called carving decomposition of the dual graph, which results in a hierarchical decomposition of $G$ into ever smaller subgraphs $H$ (ending at one face) for which the \emph{boundary} (the vertices of $H$ that may have neighbours outside $H$) has small size.   The main idea to solve {\sc PointSetEmbedding} is then to do dynamic programming in this carving decomposition, and the parameter for the dynamic program is all possible embeddings of the boundary of $H$ in the given point set $S$.

To adapt this algorithm to our situation, we need two changes.   First, we fix the point set $S$ to be the points of an $n\times n$-grid.   (Clearly no bigger grid can be required.)   In particular, we have $|S|=n^2$.   Second, when considering possible embeddings of the boundary of $H$, we \emph{only} consider such embedings where this boundary is drawn along edges of the grid with diagonals.    With this restriction, the same dynamic program will test whether a grid embedding exists in the desired time.
\end{proof}

Sadly this approach only works if the host graph is planar. Otherwise, the boundary of a subgraph does not separate its drawing from the rest.

\bibliographystyle{plainurl}
\bibliography{refs}

\begin{thebibliography}{10}

\bibitem{ALP22}
Hugo~A. Akitaya, Maarten L{\"o}ffler, and Irene Parada.
\newblock How to fit a tree in a box.
\newblock {\em Graphs and Combinatorics}, 38(5):155, 2022.
\newblock \href {http://dx.doi.org/10.1007/s00373-022-02558-z}
  {\path{doi:10.1007/s00373-022-02558-z}}.

\bibitem{BanDevDuj-Algo-19}
Michael~J. Bannister, William~E. Devanny, Vida Dujmovi{\'c}, David Eppstein,
  and David~R. Wood.
\newblock {Track layouts, layered path decompositions, and leveled planarity}.
\newblock {\em Algorithmica}, 81(4):1561{--}1583, 2019.
\newblock \href {http://dx.doi.org/10.1007/s00453-018-0487-5}
  {\path{doi:10.1007/s00453-018-0487-5}}.

\bibitem{BC87}
Sandeep Bhatt and Stavros Cosmadakis.
\newblock The complexity of minimizing wire lengths in {VLSI} layouts.
\newblock {\em Information Processing Letters}, 25(4):263--267, 1987.
\newblock \href {http://dx.doi.org/10.1016/0020-0190(87)90173-6}
  {\path{doi:10.1016/0020-0190(87)90173-6}}.

\bibitem{BV13}
Therese Biedl and Martin Vatshelle.
\newblock The point-set embeddability problem for plane graphs.
\newblock {\em Int. J. Comput. Geom. Appl.}, 23(4-5):357--396, 2013.
\newblock \href {http://dx.doi.org/10.1142/S0218195913600091}
  {\path{doi:10.1142/S0218195913600091}}.

\bibitem{BosDujJav-DMTCS-22}
Prosenjit Bose, Vida Dujmovi{\'c}, Mehrnoosh Javarsineh, Pat Morin, and
  David~R. Wood.
\newblock {Separating layered treewidth and row treewidth}.
\newblock {\em Discrete Mathematics {\&} Theoretical Computer Science},
  24(1):P18:1--P18:10, 2022.
\newblock \href {http://dx.doi.org/10.46298/dmtcs.7458}
  {\path{doi:10.46298/dmtcs.7458}}.

\bibitem{BMO22}
Prosenjit Bose, Pat Morin, and Saeed Odak.
\newblock An optimal algorithm for product structure in planar graphs.
\newblock In Artur Czumaj and Qin Xin, editors, {\em 18th Scandinavian
  Symposium and Workshops on Algorithm Theory, {SWAT} 2022, June 27-29, 2022,
  T{\'{o}}rshavn, Faroe Islands}, volume 227 of {\em LIPIcs}, pages
  19:1--19:14. Schloss Dagstuhl - Leibniz-Zentrum f{\"{u}}r Informatik, 2022.
\newblock \href {http://dx.doi.org/10.4230/LIPIcs.SWAT.2022.19}
  {\path{doi:10.4230/LIPIcs.SWAT.2022.19}}.

\bibitem{DJMMUW20}
Vida Dujmovic, Gwena{\"{e}}l Joret, Piotr Micek, Pat Morin, Torsten Ueckerdt,
  and David~R. Wood.
\newblock Planar graphs have bounded queue-number.
\newblock {\em J. {ACM}}, 67(4):22:1--22:38, 2020.
\newblock \href {http://dx.doi.org/10.1145/3385731}
  {\path{doi:10.1145/3385731}}.

\bibitem{DujMorWoo-JCTB}
Vida Dujmovi{\'c}, Pat Morin, and David~R. Wood.
\newblock {Layered separators in minor-closed graph classes with applications}.
\newblock {\em J. Combinatorial Theory, Ser. B}, 127:111{--}147, 2017.
\newblock \href {http://dx.doi.org/10.1016/j.jctb.2017.05.006}
  {\path{doi:10.1016/j.jctb.2017.05.006}}.

\bibitem{DHJLW21}
Zden{\v{e}}k Dvo{\v{r}}{\'a}k, Tony Huynh, Gwenael Joret, Chun-Hung Liu, and
  David~R. Wood.
\newblock Notes on graph product structure theory.
\newblock In Jan de~Gier, Cheryl~E. Praeger, and Terence Tao, editors, {\em
  2019-20 MATRIX Annals}, pages 513--533. Springer International Publishing,
  Cham, 2021.
\newblock \href {http://dx.doi.org/10.1007/978-3-030-62497-2_32}
  {\path{doi:10.1007/978-3-030-62497-2_32}}.

\bibitem{EW96}
Peter Eades and Sue Whitesides.
\newblock The logic engine and the realization problem for nearest neighbor
  graphs.
\newblock {\em Theoretical Computer Science}, 169(1):23–37, 1996.
\newblock \href {http://dx.doi.org/10.1016/S0304-3975(97)84223-5}
  {\path{doi:10.1016/S0304-3975(97)84223-5}}.

\bibitem{FeiHajLee-SICOMP-08}
Uriel Feige, Mohammadtaghi Hajiaghayi, and James~R. Lee.
\newblock {Improved approximation algorithms for minimum weight vertex
  separators}.
\newblock {\em SIAM J. Comput.}, 38(2):629{--}657, 2008.
\newblock \href {http://dx.doi.org/10.1137/05064299X}
  {\path{doi:10.1137/05064299X}}.

\bibitem{GJ79}
M.~R. Garey and D.~S. Johnson.
\newblock {\em Computers and Intractability: A Guide to the Theory of
  NP-Completeness}.
\newblock Freeman, 1979.

\bibitem{Gregori89}
Angelo Gregori.
\newblock Unit-length embedding of binary trees on a square grid.
\newblock {\em Information Processing Letters}, 31(4):167--173, 1989.
\newblock \href {http://dx.doi.org/10.1016/0020-0190(89)90118-X}
  {\path{doi:10.1016/0020-0190(89)90118-X}}.

\bibitem{Gri82}
David Gries.
\newblock A note on a standard strategy for developing loop invariants and
  loops.
\newblock {\em Science of Computer Programming}, 2(3):207--214, 1982.
\newblock \href {http://dx.doi.org/10.1016/0167-6423(83)90015-1}
  {\path{doi:10.1016/0167-6423(83)90015-1}}.

\bibitem{GSZ21}
Siddharth Gupta, Guy Sa'ar, and Meirav Zehavi.
\newblock Grid recognition: Classical and parameterized computational
  perspectives.
\newblock In Hee{-}Kap Ahn and Kunihiko Sadakane, editors, {\em 32nd
  International Symposium on Algorithms and Computation, {ISAAC} 2021, December
  6-8, 2021, Fukuoka, Japan}, volume 212 of {\em LIPIcs}, pages 37:1--37:15.
  Schloss Dagstuhl - Leibniz-Zentrum f{\"{u}}r Informatik, 2021.
\newblock \href {http://dx.doi.org/10.4230/LIPIcs.ISAAC.2021.37}
  {\path{doi:10.4230/LIPIcs.ISAAC.2021.37}}.

\bibitem{KT16}
Frank Kammer and Torsten Tholey.
\newblock Approximate tree decompositions of planar graphs in linear time.
\newblock {\em Theor. Comput. Sci.}, 645:60--90, 2016.
\newblock \href {http://dx.doi.org/10.1016/j.tcs.2016.06.040}
  {\path{doi:10.1016/j.tcs.2016.06.040}}.

\bibitem{MT92}
Jir{\'{\i}} Matousek and Robin Thomas.
\newblock On the complexity of finding iso- and other morphisms for partial
  k-trees.
\newblock {\em Discret. Math.}, 108(1-3):343--364, 1992.
\newblock \href {http://dx.doi.org/10.1016/0012-365X(92)90687-B}
  {\path{doi:10.1016/0012-365X(92)90687-B}}.

\bibitem{Sha-EuroCG-13}
Farhad Shahrokhi.
\newblock {New representation results for planar graphs}.
\newblock In {\em Proc. 29th European Workshop on Computational Geometry
  (EuroCG {'}13)}, pages 177{--}180, 2013.
\newblock \href {http://arxiv.org/abs/1502.06175} {\path{arXiv:1502.06175}}.

\bibitem{UWY22}
Torsten Ueckerdt, David~R. Wood, and Wendy Yi.
\newblock An improved planar graph product structure theorem.
\newblock {\em Electron. J. Comb.}, 29(2), 2022.
\newblock \href {http://dx.doi.org/10.37236/10614} {\path{doi:10.37236/10614}}.

\bibitem{WuAusPit-JAIR-14}
Yu~Wu, Per Austrin, Toniann Pitassi, and David Liu.
\newblock {Inapproximability of treewidth, one-shot pebbling, and related
  layout problems}.
\newblock {\em J. Artificial Intelligence Research}, 49:569{--}600, 2014.
\newblock \href {http://dx.doi.org/10.1613/jair.4030}
  {\path{doi:10.1613/jair.4030}}.

\end{thebibliography}

\appendix
\newpage

\section{Missing details from \cref{sec:grid}}
\label{sec:grid-details}

We first give a proof of \cref{obs:BoxVsBoxtimes}: Any graph $G$ can be modified into a graph $G'$ such that $G$ has an embedding in $P_\infty \Box P_\infty$ if and only if $G'$ has an embedding in $P_\infty {\boxtimes} P_\infty$.

\begin{proof}
    The forward direction is obvious:
    If $G \subset P_\infty \Box P_\infty$, then take the embedding, rotate it by $45^\circ$ and stretch it such that neighboring grid vertices are $5\sqrt{2}$ units apart.
    Place this in $P_\infty{\boxtimes} P_\infty$ and verify that each $T(v)$ can be placed, and for each edge of $G$ the two respective degree-$4$ can be connected as in \cref{fig:grid-graph}.

    For the other direction, assume $G'$ has an embedding in $P_\infty{\boxtimes} P_\infty$.
    Observe that for any vertex $v$ in $G$, the set $S_v = \{ w \in V(G') \colon \dist(v,w) \leq 2\}$ has size $|S_v| = 1 + 8 + 16 = 25$, and thus $S_v$ occupies a $5 \times 5$ square area $A_v$ in $P_\infty {\boxtimes} P_\infty$.
    For any edge $e=uv$ in $G$, the corresponding $5$-path $u$-$s_1$-$s_2$-$s_3$-$s_4$-$v$ must be embedded along five diagonals of $P_\infty{\boxtimes} P_\infty$ with the same slope.
    This holds as $s_2$ has four neighbors outside $S_u$ ($s_3$ and three vertices of $T(u) \setminus S_u$) and thus must be on a corner of $A_u$; and symmetrically $s_3$ lies on a corner of $A_v$.    Finally $(s_2,s_3)$ must be diagonal (and have the same slope), otherwise there would not be six vertices of $P_\infty {\boxtimes} P_\infty$ that are outside $S_u\cup S_v$ but adjacent to $s_2$ or $s_3$.
\end{proof}

Next we sketch (in \cref{fig:pathwidth-2-reduction})
how to take the specific tree from the NP-hardness construction from \cite{GSZ21}, and directly construct a tree $T'$ of pathwidth 2 that has an embedding in $P_\infty {\boxtimes} P_\infty$ if and only if $T$ has an embedding in $P_\infty \Box P_\infty$.   Thus {\sc KingGraphEmbedding} is NP-hard even for trees of pathwidth 2.

\begin{figure*}[ht]
    \centering
    \includegraphics{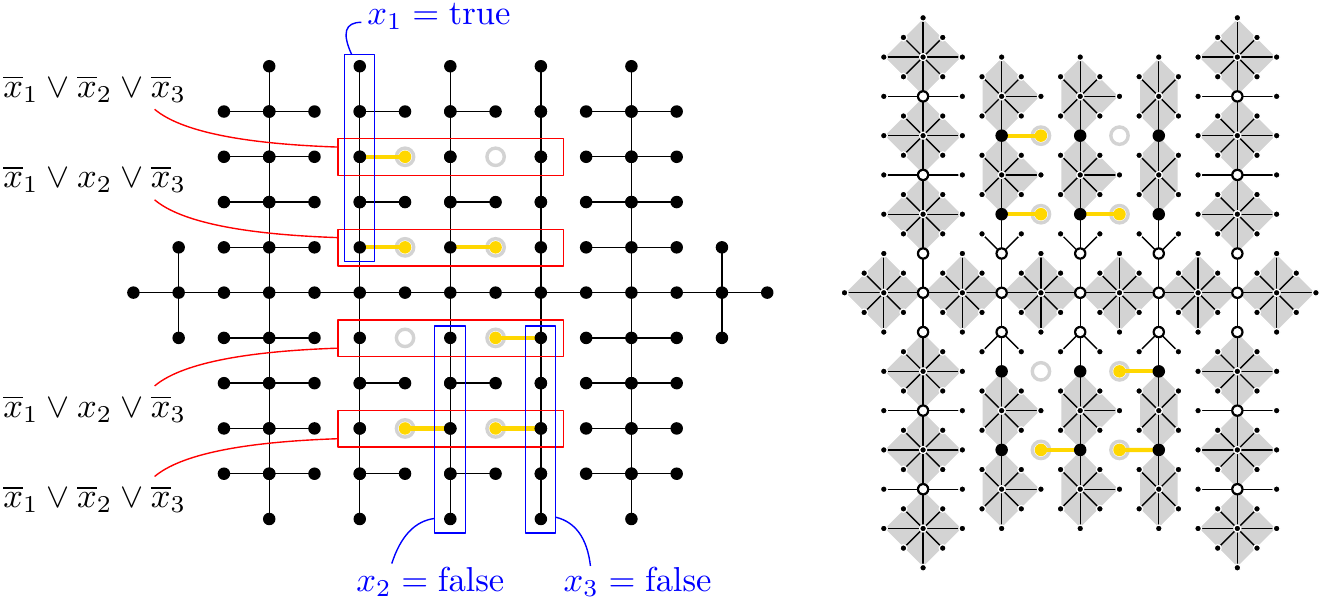}
    \caption{
        Left: The tree of pathwidth $2$ of Gupta et al.~\cite{GSZ21} for the NAE-SAT instance $\varphi = (\overline{x}_1 \vee \overline{x}_2 \vee \overline{x}_3) \wedge (\overline{x}_1 \vee x_2 \vee \overline{x}_3)$ in its {\sc GridEmbedding} for solution $\{x_1 = \text{true}, x_2 = \text{false}, x_3 = \text{false}\}$.
        Right: A corresponding tree of pathwidth $2$ in its corresponding king's graph embedding.
\todo[inline]{TB: Time permitting, I'd like to fiddle with this figure.   This is also the logic engine, so ideally the parts would use the same color-scheme as we did in the other NP-hardness proof, and the armatures and clause-rows (which now become clause-curves) should be indicated.   But I don't have energy for this right now.}
    }
    \label{fig:pathwidth-2-reduction}
\end{figure*}

\section{Row treewidth}
\label{sec:rowTW-details}

We prove here \cref{thm:row-treewidth}:
It is NP-hard to test whether a graph $G$ is a subgraph of $T{\boxtimes} P_\infty$ for some tree $T$,  even for a series-parallel graph $G$.
We already sketched the construction in Section~\ref{sec:rowTW}; we repeat the full construction here for ease of reading.

The reduction is from NAE-3SAT and uses the \emph{logic engine} by Eades and Whitesides \cite{EW96}.   We first show NP-hardness of a closely related problem.
Assume that with a graph $G$, we are also given labels `hor' and `ver' on some of its edges.  We say that an embedding of $G$ in $T{\boxtimes} P_\infty$ is 
\emph{orientation-constrained} if the edges marked `hor/ver' are horizontal
and vertical, respectively.
(Recall that horizontal/vertical means that the two endpoints of the edge have the same $T$-projection/$P$-projection.)

\begin{claim}
\label{cl:constrainedEmbedding}
Consider the following problem: `Given a graph $G$ with labels hor/ver on some edges,
does it have an orientation-constrained embedding in $T{\boxtimes} P_\infty$ for some tree $T$?'    This is NP-hard, even for a series-parallel bipartite graph $G$.
\end{claim}
\begin{proof}
Let $I$ be an instance of NAE-3SAT with $n$ variables and $m$ clauses.
We construct $G$ and at the same time discuss possible orientation-constrained embeddings of $G$ in the grid (i.e., in $P_\infty {\boxtimes} P_\infty$),
see also \cref{fig:SP_redux_1}.  (Since we restrict {\em all} edges to be horizontal or vertical, it does not matter whether the grid includes
the diagonals or not.)
Start with the \emph{frame} (orange in the figure) which consists of three paths connecting two vertices $t,b$; the \emph{middle path} has 
$H:= m+2n+1$ vertical edges, while the two \emph{outer paths} have $2n$ horizontal, $H$ vertical, and then another $2n$ horizontal edges each.
An orientation-constrained embedding of the frame in the grid is unique up to symmetry.   The middle $m$ rows of this embedding are called the \emph{clause-rows} and marked 
with one clause each.

\begin{figure*}[ht]
\hspace*{\fill}
\includegraphics[scale=0.7,page=1]{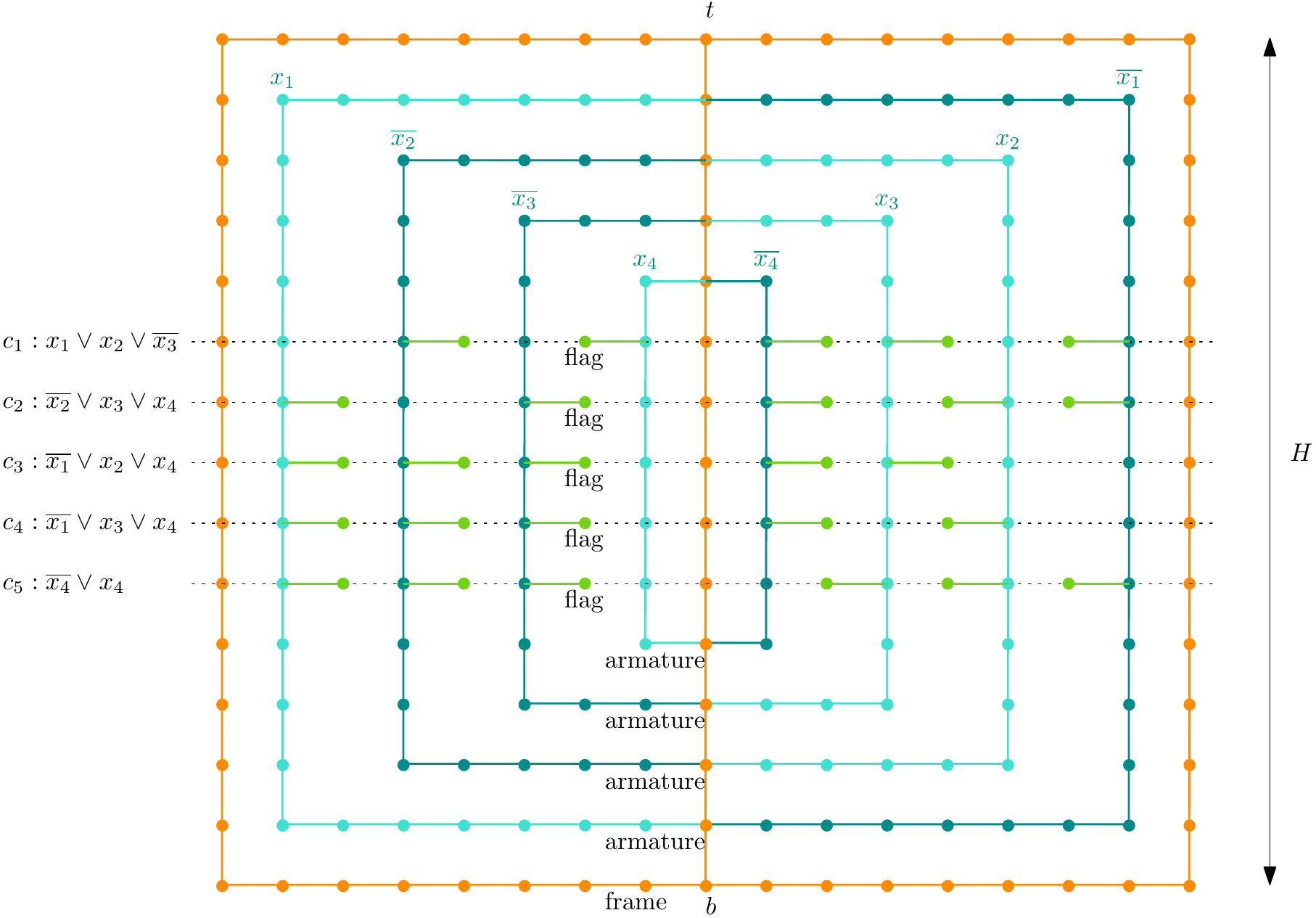}
\hspace*{\fill}
\caption{The reduction for row-treewidth   if we can fix the orientation of edges. 
}
\label{fig:SP_redux_1}
\end{figure*}

Next we add the \emph{armature of $x_i$} (light/dark cyan) for each variable $x_i$. This consists of two paths that attach at the vertices of the middle path at distance $i$ from $t$ and $b$.    Each path consist of $2n+1-2i$ horizontal edges at both ends with $H-2i$ vertical edges inbetween.
The paths are assigned to literals $x_i$ and $\overline{x_i}$.  An orientation-constrained embedding of frames and armatures in the grid is unique up to symmetry
and up to horizontally flipping each armature; in particular the row of each vertex is unchanged over all such embeddings.

Finally we attach \emph{flags} (green) at the intersections of armatures and clause-rows.   Namely, at the vertex where the armature of literal $\ell_i$ intersects the row of $c_j$, we attach a leaf (via a horizontal edge) if and only if $\ell_i$ does \emph{not} occur in $c_j$.    For each flag we have the choice of whether to place it to the right or to the left of its attachment vertex, as long as this spot has not been used by a different flag already.     Graph $G$ is clearly series-parallel, because we can reduce it to an edge by deleting leaves and multiple edges and contracting degree-2 vertices.   (We remind the reader of the following equivalent definitions of series-parallel graphs: (a) Connected graphs without a $K_4$-minor, (b) connected graphs of treewidth 2, (c) graphs obtained from an edge by attaching leaves and duplicating or subdividing edges, (d) connected graphs for which all 3-connected components contain at most three vertices.)

If $I$ has a solution, then flip the armatures such that the left paths correspond to the literals of the solution.   For each clause $c_j$ there exists at least one true literal, hence there are at most $n-1$ flags in the row of $c_j$ and left of the middle path; we can arrange them as to fit within the gaps.   There also exists at least one false literal, hence at most $n-1$ flags in the row of $c_j$ to the right of the middle path.   So we can find an orientation-constrained embedding of $G$ in the grid.   Vice versa, if we have such an embedding, then taking the literals that are left of the middle path gives a solution to $I$ because for each clause $c_j$ there must be at most $n-1$ flags on each side of the middle path, so at least one literal is true and at least one literal is false.   

So $I$ has a solution if and only if $G$ has an orientation-constrained embedding in the grid.
To finish the NP-hardness, we must argue that any orientation-constrained embedding of $G$ in $T{\boxtimes} P_\infty$ for some tree $T$ actually must reside within a grid. To see this, let $\pi$ be the path in $T$ that corresponds to the $T$-projection of one outer path of the frame.   Since the edge-orientations on the outer path are fixed, $\pi$ has length $H$ and connects the $T$-projections $t',b'$ of $t$ and $b$, so $t',b'$ have distance $H$ in $T$.   We claim that the embedding of $G$ actually resides within $\pi{\boxtimes} P_\infty$, i.e., for any vertex $v$ of $G$ the $T$-projection $v'$ of $v$ is on $\pi$.   To show this, observe that we can find a path from $t$ to $v$ by walking through the frame, then (perhaps) an armature and then (perhaps) along a flag, and always only go downward.   Similarly find a path from $v$ to $b$ that only goes downward.   The combined walk $\sigma_v$ from $t$ to $b$ via $v$ uses exactly $H$ non-horizontal edges.   The $T$-projection $\sigma_v'$ of $\sigma_v$ connects $t'$ to $b'$ and has length $H$, which by uniqueness of paths in trees implies that $\sigma_v'=\pi$ contains $v'$.   
\end{proof}

To prove \cref{thm:row-treewidth}, we take the construction of Claim~\ref{cl:constrainedEmbedding},
but add more vertices and edges to obtain a graph $G$ for which edge-orientations are
forced in any embedding of $G$ in $T{\boxtimes} P_\infty$.     

So assume that we are given an instance $I$ of NAE-3SAT.   We may assume that one clause of $I$ is $x_n\vee \overline{x_n}$, for if there is no such clause, then we can add it without affecting the solvability of $I$.    Now let $G_0$ be the graph constructed for instance $I$ as in the proof of \cref{cl:constrainedEmbedding}.
As before, $G_0$ has a unique orientation-constrained embedded in the grid up to horizontal flipping of armatures and flags, so the $y$-coordinates of vertices are fixed.   We call the vertices and edges of $G_0$ \emph{original}.


As our next step, we ``triple the width''. Roughly speaking, we insert a new column 
before and after every column that we had in the drawing of $G_0$.   Formally (and explained on the graph, rather than 
the drawing), subdivide every horizontal edge twice, and at any vertex $v$ incident to $k$ horizontal edges, attach $2-k$
leaves.   All new edges are again required to be horizontal.   
See \cref{fig:SP_redux_2}, ignoring bold lines and
diagonal edges for now. The resulting graph $G_1$
likewise has an orientation-constrained embedding in the grid if and only if the NAE-3SAT
instance has a solution.   It also clearly is series-parallel since it is obtained from $G_0$ by subdividing edges and attaching leaves.

Next we obtain $G_2$ by \emph{adding an arrow-head} at some vertical edges $vw$.  Assuming $v$ is below $w$, this means adding the edges 
 $(v_\ell,w)$ and $(v_r,w)$, where $v_\ell,v_r$ are the two neighbours of $v$ adjacent to it via horizontal edges.    
We add arrow-heads at any vertical edge of $G_1$ for which the lower endpoint $v$ does {\em not} belong to an armature. The graph stays series-parallel since each arrow-head $\{v_\ell,v,v_r,w\}$ contains a cutting pair that separates it from the rest of the graph, so adding the edges of the arrow-heads does not affect whether there are non-trivial 3-connected components.    

For the final modification we need a simple but crucial observation, which one proves by inspecting the
neighbourhood of two adjacent vertices in $T{\boxtimes} P_\infty$ for all possible orientations of the edge between them.

\begin{observation}
Let $G$ be a graph embedded in $T{\boxtimes} P_\infty$ for some tree $T$.
If $uv$ is an edge of $G$ for which $u,v$ have at least five common
neighbours, then $uv$ must be horizontal.
\end{observation}

Thus, we turn $G_2$ into $G$ by \emph{adding a $K_{2,5}$} at every horizontal edge $uv$ of $G_2$, i.e., adding
five new vertices that are adjacent to both $u$ and $v$.   This keeps the graph series-parallel and force $uv$
to be horizontal in any embedding of $G$ in $T{\boxtimes} P_\infty$.
To avoid clutter we do not show $K_{2,5}$ in \cref{fig:SP_redux_2}, but indicate it with a bold edge.   

This ends the description of our construction.   It should be straightforward to see that a solution to the NAE-3SAT instance $I$ implies that $G$ can be embedded in $C_\infty {\boxtimes} P_\infty$. Namely, we can embed $G_0$ in $\pi{\boxtimes} P_\infty$ where $\pi$ is the spine of $C_\infty$, subdivide each edge of $P_\infty$ twice to embed $G_1$, realize the arrow-heads along diagonals, and finally use 5 legs at each vertex of $\pi$ to embed the attached $K_{2,5}$'s on their $P$-extensions.   Vice versa, assume that $G$ is embeded in $T{\boxtimes} P_\infty$ for some tree $T$.   We know that all bold edges must be horizontal.   We also claim that if $vw$ was a vertical edge of $G_1$ that received an arrow-head, then the orientation of $vw$ in the embedding is vertical.   To see this, assume that the arrow-head was $\{v_\ell,v,v_r,w\}$, with $v_\ell,v_r$ connected to $v$ via horizontal edges.
Then $vw$ belongs to two triangles $\{v_\ell,v,w\}$ and $\{v_r,v,w\}$, and the two horizontal edges $(v_\ell,v)$ and $(v_r,v)$ of these triangles share endpoint $v$.   No two such triangles exist at a diagonal edge, and $vw$ cannot be horizontal since the two horizontal edges at $v$ are $vv_\ell$ and $vv_r$.   
So $vw$ is vertical.

We claim that the embedding of $G_2$ in $T{\boxtimes} P_\infty$ actually resides within $\pi{\boxtimes} P_\infty$ for some path $\pi$, i.e., in a grid.    This is argued almost exactly as in \cref{cl:constrainedEmbedding}.   Let $\pi$ be the $T$-projection of one outer path of $G_0$; since the orientations of the edges on the outer path is fixed $\pi$ has length $H$.   For any vertex $v$ of $G_2$, we can find a walk $\sigma_v$ from $t$ to $b$ via $v$ that uses exactly $H$ non-horizontal edges (they may now be diagonal). As before this implies that the $T$-projection of $v$ is also in $\pi$, so the embedding of $G_2$ is within $\pi{\boxtimes} P_\infty$, i.e., the king's graph.    As before, we can hence associate vertices of $G_2$ with points in $\mathbb{N}\times \mathbb{N}$, and speak of rows and columns of this embedding.

Since the orientations of edges on the outer paths are fixed, the drawing of the outer paths is fixed up to symmetry and spans $12n+3$ columns (including the space for the arrowheads).    The rest of $G_2$ must lie inside the outer-paths, so in particular the row of clause $x_n\vee \overline{x_n}$ (which we call the \emph{spacer-row}) has $12n+3$ points that could host vertices.    But there are three paths of the frame, $2n$ armature-paths and $2(n-1)$ flags in this row, meaning $4n+1$ original vertices use the spacer-row. Since we tripled the width, all $12n+3$ possible points in the spacer-row are used in the embedding of $G_2$.   Furthermore, the vertices in the spacer-row come as triplets connected by horizontal edges, with the middle vertex the original vertex.   Up to a translation therefore all original vertices in the space-row have $x$-coordinate divisible by 3. 
This forces any original spine-vertex $w$ to have $x$-coordinate divisible by 3 as well, because we can get from $w$ to an original vertex $v$ in the spacer-row 
using only edges that must be vertical (due to an arrow-head) or horizontal (due to a $K_{2,5}$), and the horizontal parts have length divisible by 3.   
In consequence,
all edges on the middle path of the frame must be vertical, even those that do not have an arrowhead on them.%
\footnote{This argument would be simplified if we added arrow-heads everywhere, but then the graph would not be series-parallel.}
With this, the embedding of $G$ implies an embedding of $G_1$ in $\pi{\boxtimes} P_\infty$ that is orientation-constrained, 
and we can hence extract a solution to the NAE-3SAT instance as Claim~\ref{cl:constrainedEmbedding}.
This finishes the proof of \cref{thm:row-treewidth}]

\section{Row treedepth}
\label{sec:rowTD}

%
Recall that $S_\infty$ (the infinite \emph{star})  is the tree
that consists of one \emph{center}  that is adjacent to all other vertices (the \emph{leaves}), with 
no restriction on its number of leaves.

\begin{theorem}
\label{thm:starEmbedding}
It is NP-hard to test whether a tree is a subgraph of $S_\infty{\boxtimes} P_\infty$,
even for a tree of pathwidth 2.
\end{theorem}
\begin{proof}
We use a reduction from \emph{3-partition}, where the input is a multi-set $A=\{a_1,\dots,a_{3n}\}$  that we want to split into $n$ groups that all
sum to the same integer $B=\tfrac{1}{n} \sum_{i=1}^{3n} a_i$.   
This is strongly NP-hard \cite{GJ79}, i.e., it remains NP-hard even if $A$ is encoded in unary.
We may assume that all input-numbers are multiples of 8 (otherwise multiply all of them by 8; this does not affect NP-hardness).
We describe the construction of our tree $T$ and at the same time also argue what any embedding $\Gamma$ of $T$ in $S{\boxtimes} P$ must look like.   In $S_\infty {\boxtimes} P_\infty$, we call the $P$-extension of the center $c$ the \emph{center-row}; as in \cref{obs:CvsP} we use a degree-argument to force many vertices of $T$ to be in the center-row, and finding enough space to hold all of them is the crucial idea for our reduction.  

\begin{figure*}[ht]
\centering
\includegraphics[width=0.7\linewidth,page=2]{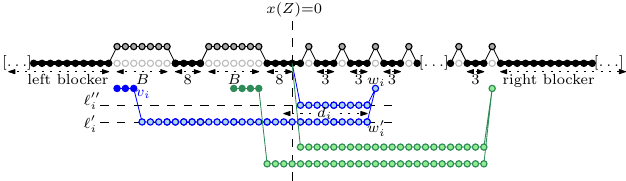}
\caption{NP-hardness of embedding in $S{\boxtimes} P$,   figure is not to scale. Filled dots represent $c$-vertices (hence have 6 leaves attached).  We only show two paddles, one green and one blue.
}
\label{fig:rowStarWidth}
\end{figure*}   

Tree $T$ consists of a \emph{frame} as well as a \emph{paddle} for each $a_i$, $i=1,\dots,3n$.   The frame is a very long path, with most vertices on the path having 6 leaves attached.   (These leaves are not shown in our picture.)   The vertices with attached leaves are called {\em $c$-vertices} and are forced to be on the central row since all other vertices of $S{\boxtimes} P$ have degree 5.  All other vertices of the frame are called {\em $\ell$-vertices} because they could be on a \emph{leaf-row} (the $P$-extension of a leaf of $S_\infty$). The specific spacing along the path is as follows:
\begin{itemize}
\item Begin with $n(B{+}8)$ $c$-vertices (the {\em left blocker}).
	Since $c$-vertices must be on the central row, and no two central-row
	vertices are adjacent unless they are consecutive, this path (and 
	similarly any path of $c$-vertices used below) occupies
	a consecutive set of vertices on the central row.
\item Continue with $B$ $\ell$-vertices, followed by $8$ $c$-vertices.
	The $\ell$-vertices could be on a leaf-row, 
	hence keep up to $B$ vertices of the central row unused.
	We call this a {\em group-gap}.   
\item We create $n$ consecutive group-gaps (in \cref{fig:rowStarWidth}, $n=2$).
\item The last vertex $Z$ of the last group-gap is called the {\em anchor}; the paddles (defined below) will attach at $Z$.
\item Starting at $Z$, we alternate between three $c$-vertices and one $\ell$-vertex that together define one \emph{fold-gap} (it permits to omit one center-row vertex).   There are $\tfrac{1}{8}n(B+8)$ fold-gaps.
\item Finally we finish with $n(B{+}8)$ $c$-vertices (the {\em right blocker}).
\end{itemize}
Note that the left and right blocker are so long that no sub-path of $\ell$-vertices could extend beyond them; in particular this forces all $c$-vertices that are not in the blockers to be between them in the central row.   

Now for each $a_i\in A$, we define the \emph{$a_i$-paddle}. This starts at anchor $Z$, continues with a path (the \emph{handle}) that has $n(B{+}8){-}1$ $\ell$-vertices, and culminates at the \emph{blade}, which consists of $a_i$ $c$-vertices. The handle is not long enough to extend beyond the blockers, so the $c$-vertices of the blade must be at $a_i\geq 2$ consecutive central-row vertices between the blockers.   Since each fold-gap leaves at most one central-row vertex free, the blade must hence occupy central-row vertices left free by a group-gap.   There are at most $nB$ such central-row vertices in $\Gamma$, and they come in blocks of at most $B$ consecutive central-row vertices each.   By $\sum_{i=1}^{3n} a_i = nB$, it follows that in any realization $\Gamma$ the group-gaps leave exactly $n$ blocks of exactly $B$ central-row vertices each, and the blades exactly fill these gaps, hence giving the desired partition of $A$.

We must still argue that if there is a solution to 3-partition, then we can embed $T$ in $S{\boxtimes} P$, and for this, need the fold-gaps and $\Theta(n)$ leaves 
for star $S$.   Embed first the frame as in the picture, so all gaps leave the maximal possible number of central-row vertices free.   (We also use 6 leaf-rows, not shown here, to embed the leaves attached at $c$-vertices.)
We treat the center-row as if it were the $x$-axis with $Z$ at the origin; this defines an $x$-coordinate $x(\cdot)$ for all embedded vertices with $x(Z)=0$.
Embed the blades of $a_1,\dots,a_{3n}$ in the group-gaps according to the solution to 3-partition.   For $i=1,\dots,3n$, 
let $v_i$ be the rightmost central-row vertex of the blade of $a_i$.    
To place the handle, we use two further leaf-rows, say $\ell_i'$ and $\ell_i''$.
We go from $v_i$ diagonally rightward to $\ell_i'$, 
then rightward for $|x(v_i)|-1$ edges to reach $x$-coordinate ${-}1$. Hence we could now go to the anchor diagonally, but the handle is longer
than this.   Therefore we continue rightward for another $d_i:=\tfrac{1}{2}(n(B+8)-|x(v_i)|)$ edges along $\ell_i'$.
Recall that each $a_i$ (and hence also $B$) is divisible by $8$.   Since there are 8 $c$-vertices at each group-gap, and all group-gaps are completely filled by paddles, $x$-coordinate $x(v_i)$ is also divisible by 8.
Thus $d_i$ is divisible by 4, and the vertex $w_i'$ that we reach
is one unit left of the central-row vertex $w_i$ of some fold-gap.   Go diagonally from $w_i'$ to $w_i$, and from there diagonally back to $x(w_i')$ on the other leaf-row $\ell_i''$.  Then we go leftward along leaf-row $\ell_i''$ to
$x$-coordinate 1 and then diagonally to $Z$.   In total we have used 
$|x(v_i)|-1+2d_i=n(B+8)-1$ vertices, which is exactly the length of the handle.
Observe that vertex $w_i$ cannot have been used by a different paddle (say the $a_j$-paddle) because $v_j\neq v_i$ are distinct central-row vertices, and 
their $x$-coordinates determine the fold-gap to be used.

Thus a solution to 3-partition gives an embedding of $G$ in $S{\boxtimes} P$ and vice versa and the problem is NP-hard.   Clearly we constructed a tree $T$; and if we removed the path $\pi$ that defined the frame then all components of $T\setminus \pi$ are either singleton-vertices (at $c$-vertices of the frame) or caterpillars (at the paddles).   Therefore $T$ has pathwidth 2.
\end{proof}

The same result also holds for embedding in $S\Box P$. We use exactly almost the same tree $T$, except at each gap of the frame the path of $\ell$-vertices is longer by two vertices and the handle-vertices have four more vertices.  Details are left to the reader.

Our constructed trees have pathwidth 2.   For a tree $T$ of pathwidth 1, the answer to `is $T\subset S_\infty {\boxtimes} P_\infty$' is trivial because the answer is always `Yes':   Such a tree is a subgraph of $C_\infty$, and $C_\infty$ can be embedded in $S_\infty {\boxtimes} P_\infty$ by placing the spine on the center-row.

\end{document}